\documentclass[a4paper,onecolumn,11pt,accepted=2017-05-09]{quantumarticle}
\pdfoutput=1
\usepackage[utf8]{inputenc}
\usepackage[english]{babel}
\usepackage[T1]{fontenc}
\usepackage{csquotes}
\usepackage{bm}

\usepackage{amsfonts,amsthm,amsmath,amssymb,mathtools}
\usepackage{dsfont,newtxtext,microtype}
\usepackage{subfigure}
\usepackage{cancel}
\usepackage{slashed}

\usepackage{booktabs,multirow}
\usepackage{braket}
\usepackage[section]{placeins}
\usepackage{float}
\usepackage{comment}
\usepackage{fullpage}
\usepackage{todonotes}
\usepackage{authblk}
\usepackage{enumerate}
\usepackage{pifont}

\usepackage[font=small]{caption}

\usepackage{sistyle}
\SIthousandsep{,}
\usepackage{cancel}

\usepackage[noorphans,vskip=0.5ex]{quoting}

\definecolor{navyblue}{rgb}{0.0, 0.0, 0.5}
\definecolor{LightPink}{rgb}{0.858, 0.188, 0.478}

\usepackage[colorlinks=true, urlcolor=blue,citecolor=purple,anchorcolor=blue, linkcolor=LightPink]{hyperref}

\newcommand{\norm}[1]{\left\lVert#1\right\rVert}

\newcommand{\HNI}{H_{\text{NI}}}
\newcommand{\ox}{\otimes}

\newcommand{\cD}{\mathcal{D}}
\newcommand{\PP}{\textsf{PP}}
\newcommand{\trho}{\tilde{\rho}}
\newcommand{\PostBQP}{\textsf{PostBQP}}
\newcommand{\pfail}{p_{\text{fail}}}

\newcommand{\jw}[1]{\textcolor{purple}{JW: #1}}

\newcommand{\ident}{\mathds{1}}

\DeclareMathOperator{\poly}{poly}

\DeclareMathOperator{\HW}{HW}
\DeclareMathOperator{\tr}{tr}

\usepackage[backend=bibtex,doi=false,isbn=false,url=false,maxbibnames=5,maxcitenames=2,sorting=none]{biblatex}
\DeclareBibliographyDriver{inproceedings}{%
  \usebibmacro{bibindex}%
  \usebibmacro{begentry}%
  \usebibmacro{author/translator+others}%
  \setunit{\labelnamepunct}\newblock
  \iffieldundef{doi}
    {\printfield[title]{title}}
    {\href{https://doi.org/\thefield{doi}}{\printfield[title]{title}}}%
  \newunit\newblock
  \printfield{booktitle}%
  \newunit\newblock
  \printlist{publisher}%
  \printfield{year}%
  \newunit\newblock
  \printfield{pages}%
  \newunit\newblock
  \usebibmacro{doi+eprint+url}%
  \usebibmacro{finentry}%
}
\DeclareBibliographyDriver{inbook}{%
  \usebibmacro{bibindex}%
  \usebibmacro{begentry}%
  \usebibmacro{author/translator+others}%
  \setunit{\labelnamepunct}\newblock
  \iffieldundef{doi}
    {\printfield[title]{title}}
    {\href{https://doi.org/\thefield{doi}}{\printfield[title]{title}}}%
  \newunit\newblock
  \printfield{booktitle}%
  \newunit\newblock
  \printlist{publisher}%
  \printfield{year}%
  \newunit\newblock
  \printfield{pages}%
  \newunit\newblock
  \usebibmacro{doi+eprint+url}%
  \usebibmacro{finentry}%
}
\bibliography{references}
\newbibmacro{string+doi}[1]{\iffieldundef{doi}{#1}{\href{https://dx.doi.org/\thefield{doi}}{#1}}}
\DeclareFieldFormat{title}{\usebibmacro{string+doi}{\mkbibemph{#1}}}
\DeclareFieldFormat[article]{title}{\usebibmacro{string+doi}{\mkbibquote{#1}}}
\DeclareFieldFormat[incollection]{title}{\usebibmacro{string+doi}{\mkbibquote{#1}}}
\usepackage[capitalise,nameinlink]{cleveref}
\Crefname{lemma}{Lemma}{Lemmas}
\Crefname{proposition}{Proposition}{Propositions}
\Crefname{definition}{Definition}{Definitions}
\Crefname{theorem}{Theorem}{Theorems}
\Crefname{conjecture}{Conjecture}{Conjectures}
\Crefname{corollary}{Corollary}{Corollaries}
\Crefname{example}{Example}{Examples}
\Crefname{section}{Section}{Sections}
\Crefname{appendix}{Appendix}{Appendices}
\crefname{figure}{Fig.}{Figs.}
\Crefname{figure}{Figure}{Figures}
\crefname{equation}{Eq.}{Eqs.}
\Crefname{equation}{Equation}{Equations}
\Crefname{table}{Table}{Tables}
\Crefname{item}{Property}{Properties}
\Crefname{remark}{Remark}{Remarks}

\newtheorem{theorem}{Theorem}
\newtheorem{definition}[theorem]{Definition}

\newtheorem{lemma}[theorem]{Lemma}

\usepackage{ltablex}

\newcommand\prob\textsc

\title{Gibbs Sampling gives Quantum Advantage at Constant Temperatures with $O(1)$-Local Hamiltonians}
\date{}

\author[1,2]{\href{https://orcid.org/0000-0001-6365-8238}{Joel~Rajakumar}}
\author[1,2]{\href{https://orcid.org/0000-0002-6077-4898}{James~D.~Watson}}

\affil[1]{ Joint Center for Quantum Information \& Computer Science, 
	National Institute of Standards \& Technology and University of Maryland, 
	College Park }
\affil[2]{Department of Computer Science and Institute for Advanced Computer Studies,
	University of Maryland, College Park}

\begin{document}

{\begingroup
\hypersetup{urlcolor=navyblue}
    \maketitle
\endgroup}
	\begin{abstract}
		Sampling from Gibbs states --- states corresponding to system in thermal equilibrium --- has recently been shown to be a task for which quantum computers are expected to achieve super-polynomial speed-up compared to classical computers, provided the locality of the Hamiltonian increases with the system size \cite{bergamaschi2024sample}.
		We extend these results to show that this quantum advantage still occurs for Gibbs states of Hamiltonians with $O(1)$-local interactions at constant temperature by showing classical hardness-of-sampling and demonstrating such Gibbs states can be prepared efficiently using a quantum computer.
        In particular, we show hardness-of-sampling is maintained even for 5-local Hamiltonians on a 3D lattice.
        We additionally show that the hardness-of-sampling is robust when we are only able to make imperfect measurements.
	\end{abstract}
	
	\section{Introduction}
    Gibbs states are fundamental objects of interest in many-body physics and chemistry, where they correspond to the state a system equilibrates to at a fixed temperature, and play an important role in semidefinite program solving and machine learning.
    One of the key uses cases for quantum computers has been to simulate many-body quantum systems, and with this in mind, a variety of quantum algorithms for Gibbs state preparation and sampling have been proposed\footnote{\cite{temme2011quantum, brandao2019finite, moussa2019low, motta2020determining,mozgunov2020completely, chen2021fast, wocjan2023szegedy, shtanko2021preparing, zhang2023dissipative, chen2023quantum, jiang2024quantum, chen2024randomized}. See table 1 of \cite{chen2023quantum} for a discussion of Gibbs state algorithms and their limitations. Beyond this, many algorithms for computing thermal expectation values by relation to the microcanonical ensemble \cite{lu2021algorithms,schuckert2023probing,ghanem2023robust}}, including a recent promising quantum generalisation of Metropolis-Hastings algorithm \cite{chen2023efficient, gilyen2024quantum}.

    At a given inverse temperature $\beta$ and Hamiltonian $H$, we define the Gibbs state as:
    \begin{align*}
        \rho(H, \beta) = \frac{e^{-\beta H}}{\tr\left[e^{-\beta H}\right]}.
    \end{align*}
    Despite their importance, the computational complexity of calculating quantum Gibbs state properties, and whether there is a quantum advantage, has remained poorly understood --- particularly at temperature independent of system size $\beta = \Theta(1)$. 
    For Gibbs states of both classical and quantum Hamiltonians, efficient algorithms are known to exist above certain critical temperatures  \cite{francca2017perfect,harrow2020classical, mann2021efficient,yin2023polynomial, rouze2024efficient, bakshi2024high}, and at sufficiently low $O(1)$ temperatures Gibbs state properties of classical Hamiltonians are known to be \textsf{NP}-hard  and even \textsf{MA}-complete to compute \cite{sly2010computational,crosson2010making, sly2012computational}.
    At even lower temperatures, $\beta = \Omega(\poly(n))$, the Gibbs state has high overlap with the system's ground state, and so for quantum Hamiltonians we can argue that cooling to these temperatures must be at least \textsf{QMA}-hard.
    Indeed, for quantum Hamiltonians at $\beta = \Omega(\poly(n))$, computing expectation values of local observables measured on Gibbs states has been shown to be hard for a class \textsf{QXC}, ``Quantum Approximate Counting'' \cite{bravyi2022quantum, bravyi2024quantum}\footnote{For $\beta=\Omega(\poly(n))$ it is also true that estimating the \textit{normalised} partition function is DQC1-hard \cite{brandao2008entanglement}}.
    But exactly how this class relates to other classical and quantum complexity classes, and whether we expect similar hardness results for $\beta=\Theta(1)$, remains poorly understood.
    The question of the complexity of Gibbs state properties is intimately related to the question of mixing times for preparing Gibbs states --- a rapid mixing time for an algorithm preparing a particular Gibbs state implies that preparing the Gibbs state in question is computationally tractable --- and many results have been proven for quantum a variety of Hamiltonians and algorithms \cite{gamarnik2024slow, basso2024optimizing, chen2024local, smid2025rapid, bergamaschi2025quantum}. 
    Gibbs sampling is also related to other important problems which have shown quantum speed-up including sampling log-concave distributions and volume estimation \cite{childs2022quantum, chakrabarti2023quantum}. 

    Recent work by \citeauthor{bergamaschi2024sample} consider the task of sampling bitstrings from quantum Gibbs states \cite{bergamaschi2024sample}. 
    That is, sampling from a distribution close in total variation distance to
    \begin{align*}
    \ P(x) = \frac{\bra{x}e^{-\beta H}\ket{x}}{\tr\left[e^{-\beta H}\right]}
    \end{align*}
    for a bitstring $x\in \{0,1\}^n$.
    \citeauthor{bergamaschi2024sample}  demonstrate that sampling bitstrings from families of Gibbs states of quantum Hamiltonians with $O(\log\log(n))$-local terms at temperature $\beta = \Theta(1)$ is classically intractable under complexity-theoretic conjectures (e.g. unless the polynomial hierarchy collapses to the third level) \cite{bergamaschi2024sample}.
    They further demonstrate that such Gibbs states can be prepared in polynomial time using a quantum computer, and hence the task of Gibbs sampling can be done efficiently with the aid of a quantum computer.
    Not only does this demonstrate a quantum advantage in Gibbs state preparation over classical computers --- and thus may be a good test of so-called quantum supremacy --- but it is arguably the first convincing demonstration that Gibbs states retain  non-trivial quantum computational properties at $\beta = \Theta(1)$ temperatures.

    In this work we build upon the work in \citeauthor{bergamaschi2024sample} and demonstrate there exist families of 5-local and 6-local Hamiltonians such that sampling from Gibbs states at $\beta = \Theta(1)$ temperatures remains classically intractable under complexity-theoretic conjectures.

    \begin{theorem}[(Informal) Classically Intractable Gibbs Sampling] \label{Theorem:Main}
    There exist two families $\mathcal{F}_1, \mathcal{F}_2$ of efficiently constructable Hamiltonians such that sampling from a probability distribution $Q(x)$ satisfying:
    \begin{align*}
        \norm{Q(x) - \frac{\bra{x}e^{-\beta H}\ket{x}}{\tr\left[e^{-\beta H}\right]} }_{1}\leq \epsilon,
    \end{align*}
    is not possible for randomised classical algorithms under complexity-theoretic conjectures, 
    for the following parameter regimes:
    \begin{enumerate}
        \item[$\mathcal{F}_1$:] the Hamiltonians are 5-local, nearest-neighbour Hamiltonians on a 3D cubic lattice, where each qubit is involved in at most $4$ Hamiltonian terms, and $\epsilon = O(\exp(-n))$.
        Furthermore, we allow each single-qubit measurement outcome to be incorrect with an $O(1)$ probability.
        
        \item[$\mathcal{F}_2$:] the Hamiltonians are 6-local, and  $\epsilon =O(1)$.
    \end{enumerate}
    Sampling from a $Q(x)$ close to either of these distributions can be done efficiently using a quantum computer.
    \end{theorem}
    
    \noindent We prove our results by combining results from the hardness of sampling IQP circuits with results from measurement-based quantum computation and error-detection protocols.
    In particular, we utilise work on the hardness of sampling from IQP circuits developed by \citeauthor{fujii2016computational} \cite{fujii2016computational} and \citeauthor{Bremner_2017} \cite{Bremner_2017} to construct families of Hamiltonians whose corresponding Gibbs states inherent this hardness of sampling.
    We anticipate that the families of Hamiltonians in $\mathcal{F}_1,\mathcal{F}_2$ are easier to realise experimentally than those presented in previous works.
    \paragraph{Paper Outline}
    In \cref{Sec:Preliminaries} we give the preliminaries and notation.
    \Cref{Sec:Hardness_of_Sampling} contains an overview of the proof techniques, as well as our main results:
    \cref{Sec:Hardness_Exponential} proves hardness of sampling Gibbs states of 5-local Hamiltonians on a 3D lattice and
    \Cref{Sec:Hardness_Polynomial} proves hardness for 6-local Hamiltonians.
 In \cref{Sec:Measurement_Errors}, we show how these hardness results can be made robust to errors in measurement.
    In \cref{Sec:Verification} we suggest a heuristic method of demonstrating that we have prepared the desired Gibbs state.
    Finally, we discuss this work and open questions in \cref{Sec:Discussion}.

\section{Preliminaries} \label{Sec:Preliminaries}

\subsection{Notation}
Consider a set of $n$ particles, each with local Hilbert space $\mathbb{C}^d$, then the full Hilbert space is $(\mathbb{C}^d)^{\ox n}$.
For a Hilbert space $\mathcal{H}$ we denote the set of bounded linear operators as $\mathcal{B}(\mathcal{H})$.
A Hamiltonian acting on $n$ qudits is a Hermitian operator $H\in \mathcal{B}((\mathbb{C}^d)^{\ox n})$.
For the rest of this work will we consider the case $d=2$, i.e. the Hamiltonian acts on qubits.
$\norm{\cdot}$ will denote the operator norm, $\norm{\cdot}_1$ will denote the Schatten 1-norm when acting on operators, and if acting on probability distributions is the standard $L1$ distance between probability distributions.

\paragraph{Hamiltonian Parameters.} The Hamiltonian is called a $k$-local Hamiltonian if we can write it as:
\begin{align*}
    H = \sum_i h_i
\end{align*}
if each $h_i$ acts on at most $k$ many qubits.
Given a local Hamiltonian, we can define a interaction (hyper)graph, where the vertices are given by the qubits and the (hyper)edges are given by the qubits $h_i$ acts non-trivially on.
We denote the maximum degree of the interaction graph as $\Delta$.
We assume that $\norm{h_i}=O(1)$ for all local terms in the Hamiltonian.

\paragraph{Thermal Physics.}

Given a temperature $T\in \mathbb{R}^+$, we will work with the inverse temperature $\beta = 1/T$.
	For a given Hamiltonian $H$ with eigenvalues $\{\lambda_i\}_i$, and a given inverse temperature $\beta$, the partition function is define as:
	\begin{align*}
		Z = \text{tr}\left[ e^{-\beta H} \right] 
		= \sum_i e^{-\beta \lambda_i},
	\end{align*}
	and the associated Gibbs state is:
	\begin{align*}
		\rho(H,\beta) = \frac{1}{Z}e^{-\beta H}.
	\end{align*}

    \paragraph{Noisy Circuits.} We will often discuss circuits by directly referring to their unitary matrix $C$, with the convention that $C\ket{0}^{\otimes n}$ is the output state after applying $C$ on $\ket{0}^{\otimes n}$. We will be interested in bit-flip noise in our circuit, and define the single-qubit channel:
    \begin{align*}
        &\cD_p(\rho) = (1-p)\rho + pX\rho X
    \end{align*}
    
	\noindent We use the following notation to refer to the output distribution of a circuit with noise on the input state:
	\begin{definition}[Circuit Sampling Distributions]
		For any circuit $C$ we will use $P_C$ to denote the output distribution of sampled bitstrings from $C$,
		\begin{align*}
			P_C(x) = |\bra{x}C\ket{0^n}|^2,
		\end{align*}
		and $P_{C,q}$ to denote the output distribution of sampled bitstrings from $C$ with independent single-qubit bit-flips applied on every input qubit with probability $q$:
  \begin{align*}
      P_{C,q}(x) = \tr\left[  \ket{x}\bra{x}\cdot C (\cD^{\ox n}_q( \ket{0}\bra{0}^{\ox n}))  C^\dagger \right].
  \end{align*}
	\end{definition}

 \subsection{Parent Hamiltonian Construction for Quantum Circuits}

	Initially consider a non-interacting Hamiltonian on $n$ qubits:
	\begin{align*}
		\HNI = \sum_{i=1}^n \frac{1}{2} (\mathds{1}-Z_i).
	\end{align*}
	It can be seen that the eigenstates of this Hamiltonian correspond to bitstrings $\ket{x}$ for $x\in \{0,1\}^n$, where the energy of eigenstate $\ket{x}$ is $\lambda_x = \HW(x)$, where $\HW$ denotes the Hamming weight of a string $x$.
	The zero-energy ground state is $\ket{0}^{\ox n}$.
	The corresponding Gibbs state of the Hamiltonian is:
	\begin{align*}
		\frac{1}{Z}e^{-\beta H} = \frac{1}{Z}e^{-\beta \sum_{x\in \{0,1\}^n} \lambda_x \ket{x}\bra{x}}
	\end{align*}
	and $Z= \sum_{x\in \{0,1\}^n} e^{-\beta \lambda_x}$. 
	We see that $Z$ is the partition function for $n$ non-interacting spins, and hence $Z= (1+e^{-\beta})^n$. 
	Now, given a circuit $C$ on $n$ qubits, we can define a \textit{parent Hamiltonian:}
	\begin{align*}
		H_C = C\HNI C^\dagger.
	\end{align*}

 \subsection{Gibbs States of Parent Hamiltonians}
 We use the following two lemmas in our work. The first (implicitly used in \cite{bergamaschi2024sample}) shows that for the set of parent Hamiltonians formed using quantum circuits, their Gibbs states are equivalent to the circuit with bit-flip noise acting on the input.
 Here the input noise strength is related to the temperature of the Gibbs state.
 We provide an explicit proof in \cref{Appendix:Gibbs_IQP_Equivalence} for completeness.
    Here the temperature of the Gibbs state corresponds to the strength of the bit-flip noise.
	\begin{lemma}\label{Lemma:Gibbs_to_Circuit_Sampling_Main}
		For any circuit $C$ constructed from $k$-local gates of depth $d$, there exists a $O(kd)$-local parent Hamiltonian $H_C$, such that:
  \begin{align*}
      \frac{1}{Z}e^{-\beta H_C} &= C(\cD_q^{\ox n}(\ket{0}\bra{0}))C^\dagger
  \end{align*}
  for $q = \frac{e^{-\beta}}{1+e^{-\beta}}$.
	\end{lemma}

 The second lemma, which is a large portion of the technical work in \cite{bergamaschi2024sample}, shows that a quantum computer can efficiently prepare the Gibbs States of these parent Hamiltonians.
 \begin{lemma}[Efficient Gibbs State Preparation for Parent Hamiltonians, Lemma 1.2 of \cite{bergamaschi2024sample}] \label{Theorem:Efficient_Gibbs_Prep}
		Fix $\beta > 0$, and let $H_C$ be the parent Hamiltonian of a quantum circuit $C$ on $n$ qubits, of depth $d$ and locality $\ell$. Then, there exists a quantum algorithm which can prepare the Gibbs state of $H$ at inverse-temperature $\beta$ up to an error $\epsilon$ in trace distance in time $O(2^{4\ell}2^d e^\beta n\poly(\log \frac{n}{\epsilon},\ell,\beta))$.
	\end{lemma}

    At a high level, this result proves a rapid mixing property of a particular Lindbladian evolution that converges to the Gibbs state. The analysis takes advantage of the fact that the Hamiltonian is the parent Hamiltonian of a $O(\log(n))$-depth circuit, which means that the `lightcones' of any qubit can be bounded.

\section{Hardness of Sampling Gibbs States with \textit{O(1)}-Local Hamiltonians}
 \label{Sec:Hardness_of_Sampling}
	\subsection{Overview of Techniques} 
 We wish to prove hardness of classically sampling the from Gibbs states of a Hamiltonian.
 To do so, we turn to \cref{Lemma:Gibbs_to_Circuit_Sampling_Main}:
 our goal is to find a family of circuits $C$ with corresponding parent Hamiltonians whose Gibbs states are hard to sample from with a classical computer but efficiently sampleable with a quantum computer. $C$ must satisfy the following requirements, 
 \begin{enumerate}
     \item For some noise strength $q = O(1)$, $P_{C,q}$ must be hard to sample from with either inverse exponential or additive error.
     \item $H_C$ must be $O(1)$-local.
     \item $C$ must have $O(\log n)$ depth.
 \end{enumerate}
 In this work, we will choose $C$ to be from the family of IQP circuits with the aim of maintaining classical hardness-of-sampling while ensuring the associated parent Hamiltonian remains $O(1)$-local.
 The main reason for choosing these circuits is that there are known constant depth IQP circuits which are hard to sample from Refs. \cite{fujii2016computational, Hangleiter_2018}, and the shallowness of the circuit is useful to obtain $O(1)$-locality of the parent Hamiltonian. 
 In particular, we need to choose a family of IQP circuits whose hardness of sampling is robust to bit-flip noise on the input.
 This is non-trivial since in many cases the presence of noise may make the IQP circuit classically easy to sample \cite{bremner2016average, rajakumar2024polynomial}.
 
    To prove our results, we use two separate error mitigation techniques. 
    First, we use a result from \citeauthor{fujii2016computational} \cite{fujii2016computational}.
    \citeauthor{fujii2016computational} make use of topologically protected circuits to ensure that, provided the circuit noise is below a certain threshold, exact sampling from the circuit remains difficult. 
    Modifying this proof allows us to obtain hardness of sampling with inverse exponential error, and 
    this allows us to arrive at $\mathcal{F}_1$ of \cref{Theorem:Main}, a family of $5$-local Hamiltonians with degree $5$, which are hard to sample from with inverse exponential error. 
    
    Second, we implement an error-detection strategy inspired by the techniques of Refs.
    \cite{Bremner_2017, bergamaschi2024sample}, where parts of the circuit are repeated and we simultaneously introduce flag qubits which flip to show where an error may have occurred.
    This allows us to arrive at hardness of sampling from with additive error for $\mathcal{F}_2$, a family of Hamiltonians with $O(1)$ locality.
    Our hardness-of-sampling proof follows the similar outline as Ref. \cite{bergamaschi2024sample}, but here we prove hardness for different sets of circuits which allows us to obtain hardness for families of Hamiltonians with $O(1)$-locality rather than $O(\log\log(n))$-locality. 
    We also note similar work on thermal states of measurement-based thermal states in \cite{fujii2015quantum}.
 
    Finally, by combining these two techniques, we obtain a family of $6$-local Hamiltonians with max degree $\Delta$, such that their Gibbs states are hard to sample from with inverse exponential error when $\beta = O(1/\sqrt{\Delta})$, which complements easiness results from \cite{bakshi2024high}.
    \color{black}

	\subsection{Hardness of Sampling Gibbs States on 3D Lattices}
 \label{Sec:Hardness_Exponential}

    \citeauthor{fujii2016computational} demonstrate a family of constant-depth, geometrically local IQP circuits which are hard to \textit{exactly} sample from in the presence of noise of constant strength \cite{fujii2016computational}. 
    At a high level, \citeauthor{fujii2016computational} give a protocol in which a cluster state is constructed using 
    a depth 4 IQP circuit. 
    If the noise level is sufficiently smaller than the threshold value for topologically protected MBQC, then the output distribution of the circuit cannot be exactly sampled from. 
     Provided the noise is restricted to bit-flip noise and kept below a threshold strength of $\sim 0.134$, then T-gate synthesis is still possible in the circuit,
     and postselecting on error-free outcomes allows one to decide $\textsf{PostBQP}$-complete problems. 
     Since $\textsf{PostBQP}=\textsf{PP}$ and it is believed that $\textsf{PostBQP}\not \subseteq\textsf{PostBPP} $ unless the polynomial hierarchy collapses to the third level, then sampling from the noisy circuit is still hard \cite{aaronson2005quantum}.

	By taking this set of topologically protect circuits from \citeauthor{fujii2016computational} we prove the following lemma in  \cref{Appendix:Hardness_Exponential}.
	\begin{lemma}  \label{Lemma:Hardness_Exponential}
		There exists a family of IQP circuits $\mathcal{C}$, constructed on a 3D cubic lattice, consisting of a single layer of $e^{iZ\pi/8}$ gates and 4 layers of nearest-neighbour $e^{iZZ\pi/4}$ gates, such that sampling from any probability distribution $Q$ over bitstrings which satisfies $\|P_{C,q}-Q\|_{1} < (1-q)^n2^{-6n-4}/5$ is not possible with a classical polynomial algorithm, unless the polynomial hierarchy collapses to the third level, when $q < 0.134$.
	\end{lemma}
 
	\noindent
    We now want to use this family of classically hard-to-sample circuits to generate hard-to-sample Gibbs states.
    In particular, we combine this hard-to-sample from family of IQP circuits with the fact that the Gibbs states of IQP parent Hamiltonians look like the output of noisy IQP circuits, as per \cref{Lemma:Gibbs_to_Circuit_Sampling_Main}, and we get the following.
 
	\begin{theorem}[Classical Hardness of Gibbs Sampling] \label{Theorem:4Local}
		One can efficiently construct a family of $5$-local Hamiltonians on a 3D cubic lattice of qubits, $\{ H_C\}_C$, whose Gibbs states at $\beta = O(1)$ are hard to classically sample from with $1$-error $(1-q)^n2^{-6n-4}/5$, unless the polynomial hierarchy collapses to the 3rd level.
	\end{theorem}
	\begin{proof}
		We choose a hard-to-sample IQP circuit $C$ from \cref{Lemma:Hardness_Exponential} and specify the new Hamiltonian as $H_C=C\HNI C^\dagger$. By \cref{Lemma:Gibbs_to_Circuit_Sampling_Main}, the output distribution of the Gibbs states of $H_C$ is exactly $P_{C,q}$
		Choosing:
		\begin{align} \label{Eq:Hardness_Condition}
			q = \frac{e^{-\beta}}{1+e^{-\beta}} \leq 0.134
		\end{align}
		and we see that the output distribution must be hard to sample from, as per \cref{Lemma:Hardness_Exponential}.
		Thus we see that for $\beta \gtrsim 1.87 =  O(1)$, it is hard to sample this distribution. 
  
    To see that the locality of $H_C$ is $5$, we note that for each qubit $i$, there are at most $4$ gates acting on qubit $i$ in $C$, each of which introduces interaction with at most $4$ other qubits. Thus, when conjugating the $Z_i$ term of $H_{NI}$ with the gates of $C$, it becomes an at most $5$-local term.
	\end{proof}

      Notably, in order to prove hardness of sampling with inverse exponential error from a noisy IQP circuit, a postselection argument suffices, rather than a full fault-tolerant encoding.  
      Thus, there is no added overhead needed to give hardness of sampling, which ensures that the locality of the parent Hamiltonian is purely determined by the logical IQP circuit structure, which has low depth and locality. This is not true for the circuits and parent Hamiltonians in the next section, or in Ref. \cite{bergamaschi2024sample} where additional elements need to be added to the circuit to allow for error detection.

\subsection{Hardness of Gibbs Sampling to \textit{O}(1/poly(\textit{N})) Error}
\label{Sec:Hardness_Polynomial}

 While the hardness of sampling from noiseless IQP circuits with inverse exponential error can be obtained through a postselection argument, proving hardness of sampling with additive error typically requires a few more ingredients (e.g. anticoncentration, average-to-worst case reduction). Here, we use existing results from \citeauthor{Hangleiter_2018} \cite{Hangleiter_2018} which construct a family of constant-depth IQP circuits that exhibit such hardness (up to some complexity-theoretic conjectures).
\begin{lemma} (Corollary 12 of \cite{Hangleiter_2018}) \label{Lemma:HangleiterCircuit}
There exists a family of constant depth IQP circuits $\{C_n\}_n\geq 1$ on a 2D square lattice of $n$ qubits, such that no randomised classical polynomial-time algorithm can sample from the output distribution of $C_n$ up to additive $1$ error of $\delta=1/192$, assuming the average-case hardness of computing a fixed family of partition functions, and the non-collapse of the polynomial hierarchy.
\end{lemma}    

Now we would like to show that the hardness of sampling from this circuit is robust to input noise, and without postselection (as this would make it more difficult to show hardness of additive error sampling). 
This task has been explored in Refs. \cite{Bremner_2017, bergamaschi2024sample} for IQP circuits, and they each provide a method of encoding an arbitrary IQP circuit into a larger, noise-robust circuit. 
Here, we describe an encoding method that is heavily inspired by these techniques.

Our construction is as follows. Suppose the initial circuit acts on $n$ qubits. 
Fix some positive integer $r$. 
For each `logical' qubit $i$ in the circuit of \cref{Lemma:HangleiterCircuit} (initialised in the $\ket{0}$ state), initialise a block of $r$ `physical' qubits in the $\ket{0}$ state, and label them $i_1,\ldots,i_r$. 
For each block $i$, apply $r-1$ CNOT gates which are all controlled on $i_1$, and targeting a qubit in $i_2,\ldots,i_r$. 
Let the unitary corresponding to this CNOT network be $B$. Then, apply the logical circuit (from \cref{Lemma:HangleiterCircuit}) amongst the first qubits of each block (qubits labelled $1_1,2_1,\ldots, n_1$). 
Finally, measure all qubits in the computational basis. See \cref{fig:IQP_Encoding} for an example of such a circuit with $r=2$.

\begin{figure}[h!] 
    \centering
    \includegraphics[width=1\textwidth]{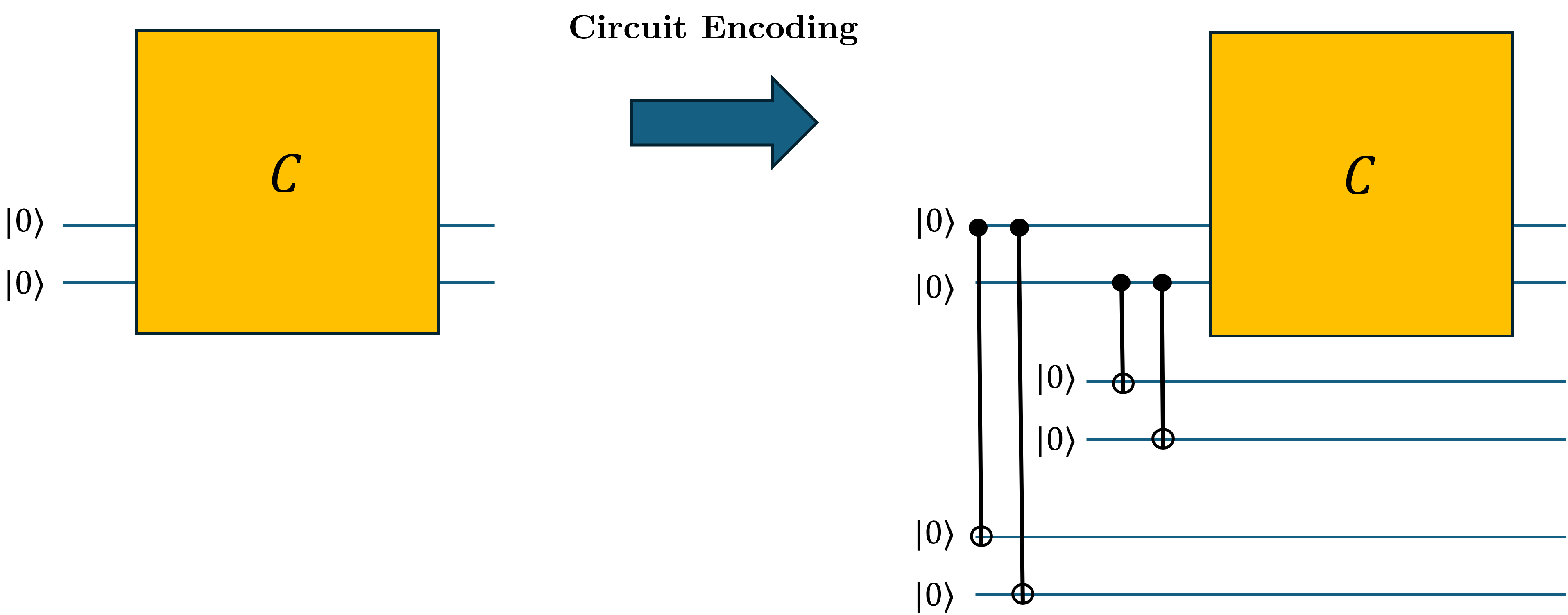}
    \caption{The encoding of the IQP circuit $C$ with $r=3$. For clarity, only two input qubits are shown explicitly, but the procedure applies to all qubits present.}
    \label{fig:IQP_Encoding}
\end{figure}

We see that if the initial logical qubit is flipped by noise at the start of the circuit, the CNOT network associated with that qubit should flip all the associated ancillary qubits to the $\ket{1}$ state, thus flagging the error which can be corrected in postprocessing.
We require $r$ auxiliary qubits to flag the error as the auxiliary qubits themselves may have noise incident on them, and so the multiple copies acts a repetition code to suppress this case.

We show in \cref{Appendix:CNOT_Encoding} that with classical post-processing, one can recover the logical output distribution with an error rate that is exponentially suppressed in $r$. 
Crucially, we only require $r=O(1)$ which means the parent Hamiltonian remains $O(1)$-local. 
Our proof works by examining the propagation of Pauli's through the CNOT network.

\begin{lemma} \label{Lemma:CNOT_Encoding_Main}
Let $C$ be an arbitrary IQP circuit constructed with $2$-qubit gates of depth $d$ on $n$ qubits. Then, for any integer parameter $r \geq 1$, there is an encoded IQP circuit $C^*$ constructed with $2$-qubit gates of depth $d+r$ on $nr$ qubits, and a decoding algorithm $A^*$ such that,
\begin{align}
    A^*(P_{C^*,q}) = P_{C,\pfail}
\end{align}
where $\pfail \leq (4q(1-q))^{r/2}$. Furthermore, the parent Hamiltonian of $C^*$ has locality $k \leq d+2$ and degree $\Delta \leq r(d+1)$
\end{lemma}

Thus, given an initial circuit $C$, obtain a new circuit composed of an initial set of CNOT layers, followed by an IQP circuit.
The output of the error-free circuit can be extracted from classical post-processing.
We can put this all together to obtain the following theorem. 

\begin{theorem}\label{Theorem:Constant_Precision_Hardness}
    One can efficiently construct a family of $6$-local Hamiltonians, $\{ H_C\}_C$, whose Gibbs states at $\beta = O(1)$ are hard to classically sample from with constant $1$-error up to $1/192$, under the assumptions of \cref{Lemma:HangleiterCircuit}.
\end{theorem}
\begin{proof}
    We choose a hard-to-sample IQP circuit $C$ from \cref{Lemma:HangleiterCircuit}. 
    Then, we encode $C$ into $C^*$ as in \cref{Lemma:CNOT_Encoding_Main} (with parameter $r$ to be set later), and specify the new Hamiltonian as $H_{C^*}= C^*\HNI C^{*\dagger}$. 
    By \cref{Lemma:Gibbs_to_Circuit_Sampling_Main}, the output distribution of the Gibbs states of $H_{C^*}$ at constant temperature is exactly $P_{C^*,q}$. 
    Using $A^*$, we can postprocess this to obtain $P_{C,\pfail}$ where $\pfail \leq (4q(1-q))^{r/2}$.
    Therefore, it suffices to choose  $r = O(\log n)$ in order to set the probability of an error happening on any qubit to be $O(1/n)$. By increasing $r$, we can set the approximation error arbitrarily lower than $1/192$. Thus, it is hard to sample from the Gibbs states of $H_{C^*}$ with constant $1$ error up to $\delta = 1/192$.
\end{proof}

Additionally, in \cref{Appendix:CNOT_Encoding} we show that the output distributions of this work and previous noise-robust IQP encodings (specifically those presented in here and in Refs. \cite{Bremner_2017, bergamaschi2024sample}) are equivalent, up to reversible classical postprocessing. 
This highlights the fact that improvements in Hamiltonian locality have not been obtained by considering very different output distributions, but rather, different Hamiltonian constructions.

\subsection{Quantum Advantage from Gibbs Sampling of \textit{O}(1)-Local Hamiltonians}

So far we have proven hardness of sampling from Gibbs states for the families of parent Hamiltonians in \cref{Sec:Hardness_Exponential} and \cref{Sec:Hardness_Polynomial}.
Here we employ the Gibbs state preparation algorithm of Ref. \cite{chen2023quantum} in conjunction with \cref{Theorem:Efficient_Gibbs_Prep}.

\begin{theorem}[Quantum Advantage with $O(1)$-Local Hamiltonians]
    The Gibbs states associated with the families of parent Hamiltonians in \cref{Theorem:4Local} and \cref{Theorem:Constant_Precision_Hardness} can be sampled from efficiently using a quantum computer in the same parameter regimes for which they are classically hard.
\end{theorem}
\begin{proof}
    For both of the families of Hamiltonians in \cref{Theorem:4Local} and \cref{Theorem:Constant_Precision_Hardness} the Hamiltonians are $O(1)$-local, hence the runtime of the Gibbs state preparation algorithm in \cref{Theorem:Efficient_Gibbs_Prep} with $\beta=\Theta(1)$, $\ell=O(1)$ and 
 $\epsilon=\frac{1}{9}(1+e^{-\beta})^{-n}2^{-6n-4}$ gives a total runtime of $O(2^{4\ell}2^d e^\beta n\poly(\log \frac{n}{\epsilon},\ell,\beta)) = O(\poly(n))$.
 Thus sampling from these distributions can be done in polynomial time for a quantum computer.
 
 On the other hand, as shown in \cref{Theorem:4Local} and \cref{Theorem:Constant_Precision_Hardness}, sampling from these Gibbs states in this parameter regime is classically intractable under complexity-theoretic conjectures.
\end{proof}

\section{Measurement Errors in Gibbs State Sampling}
\label{Sec:Measurement_Errors}

 When trying to show a quantum advantage in Gibbs state sampling, there are two sources of error: \textit{(a)} we can only approximately prepare the true Gibbs state, and \textit{(b)} when sampling, we expect our measurements to have an error associated with them.
	Here we show that, even with these two sources of error, we are sampling from a probability distribution that we do not expect to be able to sample from classically.
 \begin{theorem}
		Consider a family of IQP circuits $\mathcal{C}$, the associated parent Hamiltonians $\{H_C\}_{C\in\mathcal{C}}$ and an inverse temperature $\beta=\Theta(1)$. 
		There are a family of efficiently preparable states $\{\trho_C\}_{C\in \mathcal{C}}$ in time $\poly(n)$ on a quantum computer, such that for $\epsilon =\frac{1}{9} (1+e^{-\beta})^{-n}2^{-6n-4}$,
		\begin{align*}
			\norm{\trho_C - \rho(H_C,\beta)}_1\leq \epsilon
		\end{align*}
		and sampling from $\cD_q^{\otimes n}(\trho_C)$ (i.e. with imperfect measurements with single qubit measurement error $q = \Theta(1)$) is not possible using a randomised classical algorithm unless the polynomial hierarchy collapses to the 3rd level.
	\end{theorem}
	\begin{proof}
		From \cref{Theorem:Efficient_Gibbs_Prep}, we see that preparing $\trho_C$ to this precision is possible in $O(\poly(n))$ time using a quantum computer. The distribution we are sampling corresponding to faulty measurements with probability $q$ on the prepared state $\trho$ approximating $\rho(H_C, \beta)$ is:
		\begin{align*}
			\tilde{P}(x) = \tr\left[ \ket{x}\bra{x}\cdot \cD_q^{\ox n}(\trho_C)  \right].
		\end{align*}
		We using H\"older's inequality:
		\begin{align*}
			\norm{\tilde{P}(x) - \tr\left[ \ket{x}\bra{x} \cD_q^{\ox n}(\rho(H_C,\beta))  \right]}_{1} 
			&\leq \norm{ \cD_q^{\ox n}(\trho_C) -   \cD_q^{\ox n}(\rho(H_C,\beta))}_1.
		\end{align*}
		Since $\norm{\trho_C - \rho(H_C,\beta)}_1 \leq \epsilon$, then we see:
		\begin{align*}
			\norm{ \cD_q^{\ox n}(\trho_C -   \rho(H_C,\beta))}_1 &\leq \norm{ \trho_C -   \rho(H_C,\beta)}_1 \\
			&\leq \epsilon,
		\end{align*}
		where we have used that $\cD^{\ox n}(\cdot)$ is a CPTP map.
		Thus we see that sampling from $\tilde{P}(x)$ approximates sampling from the Gibbs state $\cD^{\ox n}(\rho(H_C,\beta))$ with imperfect measurements.
  
		Now, we show that $\cD^{\ox n}(\rho(H_C,\beta))$ is also hard to classically sample from. 
        We know from \cref{Lemma:Gibbs_to_Circuit_Sampling_Main} that we can write Gibbs states of parent Hamiltonians as equivalent to the output of a noisy IQP circuit:
		\begin{align*}
			\cD^{\ox n}_q\left(\frac{e^{-\beta H_C}}{Z} \right)&= \cD^{\ox n}_{q} \cD^{\ox n}_{\frac{e^{-\beta}}{1+e^{-\beta}}}(C\ket{0}\bra{0}^{\otimes n}C^\dagger) \\
            &= \cD^{\ox n }_{q'}(C\ket{0}\bra{0}^{\otimes n}C^\dagger)\\
            &= \frac{e^{-\beta' H_C}}{Z} 
		\end{align*}
	  where we have defined $q' =(\frac{e^{-\beta}}{1+e^{-\beta}})(1-q) + q(1-\frac{e^{-\beta}}{1+e^{-\beta}})$ and $ q' = \frac{e^{-\beta'}}{1+e^{-\beta'}}$. 
    
		Provided we keep $\beta'$ sufficiently small (but still $\Theta(1)$), we see that sampling from $\rho(H_C,\beta')$ is classically intractable from \cref{Theorem:4Local}, hence $\cD^{\ox n}_q(\rho(H_C,\beta))$ is classically hard to sample too.
    In particular, we can choose values of $\beta, q = \Theta(1)$ simultaneously.
	\end{proof}

  \section{Heuristic Verification Procedure }
 \label{Sec:Verification}

 One the of the major limitations with many quantum supremacy experiments is that verifying that the desired experiment has actually been implemented --- and not ruined by noise or an adversary --- is not trivial \cite{hangleiter2019sample}.
 For example, in the case of Random Circuit Sampling and IQP sampling, cross-entropy benchmarking is often used as a proxy for measuring fidelity between the actual state and the ideal state \cite{boixo2018characterizing,hangleiter2024fault}. 
 However, although cross-entropy benchmarking only requires a polynomial number of samples, it requires computing samples of the ideal distribution which is exponentially expensive, and it appears to be spoofable for many classes of circuits.
 Thus verifying that one is sampling from a distribution close one taken from random circuits is highly non-trivial.

\paragraph{A Heuristic Verification Protocol} 
Suppose we wish to implement the Gibbs sampling procedure using the Hamiltonians in this work as a quantum supremacy test, then we have a similar problem --- how do we verify that we are sampling from the correct distribution?
 Here we suggest a heuristic test for correctness --- we ask to verify that the Gibbs state we are sampling from has the correct Hamiltonian.
 That is, suppose for part of our sampling routine, we wish to sample from the state $\rho(H_P,\beta)$, where $H_P$ is a Hamiltonian formed from a sum of Pauli operators and can be written as:
 \begin{align*}
     H_P = \sum_{i=1}^m \mu_i P_i
 \end{align*}
 where each $P_i$ is a $k$-local Pauli string. 
 To verify that we are correctly preparing Gibbs states, we take $N$ copies of the Gibbs state $\rho(H_P,\beta)$.
 We then use the Hamiltonian learning algorithm proposed in \cite[Theorem 6.1]{bakshi2023learning} to learn the coefficients of the Hamiltonian.
 In particular, if we want to learn the parameters to precision $(\tilde{\mu}_i-\mu_i)^2\leq \epsilon$ with probability $>1-\delta$, we need a number of sample $N$:
 \begin{align*}
    N= O\left( \frac{m^6}{\epsilon^{c}}+ \frac{c}{\beta^2\epsilon^2}\log\left(\frac{m}{\delta}\right)\right),
 \end{align*}
 for a constant $c$ which depends on the locality and maximum degree of the interaction graph.
 Then we check closeness for all parameters $\mu_i$ which appear in our Hamiltonians.
 We can then use the following lemma to argue that if the measured Hamiltonians are close, then the sampled states are close:
 \begin{lemma}[Lemma 16 of ~\cite{brandao2017quantum}]\label{lem:Gibbs_perturbation}
		Let $H,H'$ be Hermitian matrices. Then:
		\begin{align}
			\norm{\frac{e^{-H_1}}{\tr\left[e^{-H_1}\right]}-\frac{e^{-H_2}}{\tr\left[e^{-H_2}\right]}}_1 \leq 2(e^{\|H_1-H_2\|_\infty}-1).
		\end{align}
	\end{lemma}
 \noindent In particular, if $\beta\|H_1-H_2\|\leq \epsilon \ll 1$, then $\norm{\rho(H_1,\beta)- \rho(H_2,\beta)}_1\leq O(\epsilon)$. 

 In particular, if we wish to verify that the actual observed Hamiltonian is close to the ideal Hamiltonian on all $k$-local terms, then this reduces to learning a Hamiltonian with $4^k\binom{n}{k}=O(\poly(n))$-many terms.
 This can be done efficiently using the algorithms of \cite{bakshi2023learning}.

\paragraph{Limitations of the Protocol}
 However, our measurement approach is limited in the sense that the state prepared may not have been the Gibbs state of the desired form, but some other state $\rho$.
 We can always write $\rho = \frac{1}{Z}\sum_i p_i \ket{\psi_i}\bra{\psi_i} = \frac{1}{Z}\sum_i e^{-\beta \lambda_i} \ket{\psi_i}\bra{\psi_i}$ where we have written $p_i = \frac{e^{-\beta \lambda_i}}{Z}$ for some ``false'' energies $\lambda_i$.
 This gives and effective Hamiltonian $H_{\text{eff}} = \sum_i \lambda_i \ket{\psi_i}\bra{\psi_i} = \sum_j \mu'_j P_j$.
 Thus if the state is incorrect, and not close to the Gibbs state, then when we apply our learning algorithm we will learn the parameters of $H_{\text{eff}}$.
 Although the $\mu'_i$ associated with $H_{\text{eff}}$ may be close to $H_P$ on all the parameters $\mu_i$, there may be some $\mu'_i$ which are non-zero (and large) when $\mu_i$ is zero.
 Thus $H_{\text{eff}}$ will look close to $H_P$ for the Pauli coefficients that we have measured.
 In general, to distinguish $H_{\text{eff}}$ and $H_P$, we may have to make $4^n$ many measurements.
 
 Despite this, it remains an open question whether Hamiltonian structure learning can be done efficiently for quantum Gibbs states.
 That is, can we not only estimate the values of the parameters $\mu_i$, but determine the entire Hamiltonian structure?
 Classically, this can be done efficiently \cite{vuffray2016interaction,klivans2017learning}.

Finally, we note that although the above Hamiltonian learning procedure may be spoofed if the measurements are performed by an untrustworthy source, they may be useful for verification in a laboratory setting where measurements are trusted.

\paragraph{Gibbs Sampling as a Test for Quantum Advantage}

We briefly note here that the algorithms considered in this work for Gibbs state preparation are well beyond the power of NISQ devices, and likely require a full fault tolerant quantum computer to implement.
As such, we believe that Gibbs sampling as a test of quantum advantage may be more appropriate for analogue or dissipative models of computation which naturally approach a Gibbs state in certain conditions. 
Alternatively, we suggest that, for the specific case of IQP parent Hamiltonians, there may be much simpler algorithms which guarantee convergence to the Gibbs state.
Furthermore, because Gibbs states are fixed points of Lindbladians, we might expect a certain level of robustness to external noise.
For example, one might hope that Trotterizing a parent Hamiltonian and running its time-evolution on a weakly-noisy device may result in the state converging to the Gibbs state.
 
\section{Discussion and Future Work}
\label{Sec:Discussion}

In this paper we have constructed two families of $O(1)$-local Hamiltonians for which their corresponding Gibbs states are classically hard to sample at $O(1)$ temperatures.
Furthermore, these Gibbs states can be efficiently prepared and sampled from using a quantum computer.
By showing that our hardness-of-sampling results hold for $\beta = \Omega(1/\sqrt{\Delta})$, we have placed bounds on where classical sampling algorithms can be efficient.
Finally, we have suggested some heuristic routes for verifying quantum advantage via Gibbs sampling.

Beyond the work studied here, there exist many potential future routes for improvement.
	\begin{enumerate}
 
	    \item 
	
	Although the Hamiltonians here have constant locality, they do not necessarily correspond to Hamiltonians seen in nature, or to distributions we are interested in sampling from in computer science. 
	A natural question to ask it whether sampling from Gibbs states of geometrically local Hamiltonians is still hard at constant temperature, even considering hardness with additional constraints such as translationally invariant terms or certain symmetries or restricted families of interaction terms.
	Remarkably, in the case of ground states, predicting local observables is known to be hard on 2D lattice and 1D translationally invariant chains \cite{gottesman2009quantum,gharibian2019oracle, watson2020complexity}.
    While the properties of Gibbs states in 1D are well understood to be easy-to-sample, for 2D and beyond is not.
	
	\item We also note that the parent Hamiltonian here has same partition function as a non-interacting Hamiltonian.
	Since non-interacting Hamiltonians do not undergo thermally-driven phase transitions, the parent Hamiltonian here does not either.
	Yet there is an apparent transition in temperature where the system moves from easy-to-sample to hard-to-sample.
	It would be interesting to understand if this transition coincides with other physical changes in the system, or if there is some way of characterising the change in sampleability in terms of other physics properties such as entanglement, entropies, etc.

    \item The bounds on efficient sampling of Gibbs states by a classical computer leaves a gap open relative to the boundary showing classical easiness of sampling in \citeauthor{bakshi2024high} \cite{bakshi2024high}, i.e. $O(\sqrt{\Delta})$ vs $\Omega(\Delta)$. 
    There is future work to be done to understand where exactly this boundary lies.

    \item Recently there has been much discussion of learning Hamiltonians with sample access to Gibbs states \cite{anshu2021sample, anshuweb, haah2022optimal, alhambra2023quantum, onorati2023efficient, rouze2024learning, garcia2024estimation, bakshi2024high}.
    Universally, the sample complexity of these algorithms increases as the temperature drops.
    However, it is notable that some of these algorithms are only efficient in the high temperature regime, and their costs blows up exponentially past some critical threshold.
    It would be desirable to understand how these thresholds relate the onset of sampling complexity (if there is any relation at all).

    \item As far as we are aware, classical hardness-of-sampling results for noisy IQP circuits from Ref. \cite{fujii2016computational} are only known for sampling to exponentially high precision in total variation distance (i.e. multiplicative error). 
    It remains an open question whether the output distributions of these noisy circuits can be shown to anti-concentrate. 
    If they do anti-concentrate, then this would imply hardness of classical sampling to $1/\poly(n)$ error in total variation distance, which then implies hardness of sampling Gibbs states with $O(1/\poly(n))$ error for the circuits in \cref{Sec:Hardness_Exponential}.

 \end{enumerate}

	\section*{Acknowledgements}

 {\begingroup
		\hypersetup{urlcolor=navyblue}

  The authors thank \href{https://orcid.org/0009-0008-5983-9130}{Atul Mantri} and \href{https://orcid.org/0000-0002-4766-7967}{Dominik Hangleiter} for useful discussions about IQP circuits.
 The authors also thank \href{https://orcid.org/0000-0001-7458-4721}{Yi-Kai Liu} and  \href{https://orcid.org/0000-0001-9743-7978}{Anirban Chowdhury} for useful feedback and discussions, as well as \href{https://orcid.org/0000-0002-7451-9687}{Ewin Tang} for correcting some of our errors in the initial manuscript.
		\endgroup}

    JR acknowledges support from the National Science Foundation Graduate Research Fellowship Program under Grant No. DGE 1840340.
    JDW acknowledges support from the United States Department of Energy, Office of Science, Office of Advanced Scientific Computing Research, Accelerated Research in Quantum Computing program, and also NSF QLCI grant OMA-2120757.

 \section*{Author Contribution Statement}

 Both authors contributed equally and are listed alphabetically.
	
	
		\printbibliography[heading=bibintoc]
	
	\appendix
	\section*{Appendix}
	
	\section{Hardness of Noisy IQP Sampling to Inverse Exponential Approximation} \label{Appendix:Hardness_Exponential}
	We make modifications to the proof from Ref. \cite{fujii2016computational} to show that noisy IQP circuits are hard to sample with, even with inverse exponential error. First, we introduce the following lemma (derived in a different form in Ref. \cite{fujii2016computational}) which allows us to bound the probability that postselected probabilities are close.
	
	\begin{lemma}  \label{Lemma:postselect}
		Let $P$ and $P'$ be any two probability distributions over bitstrings of length $n$. Suppose further that the bitstrings start with a decision bit $x$, and are followed by a `postselection' register $y$ of $n-1$ bits. For any $0 < \delta < 1/2$, if $\|P'-P\|_{1} < \frac{\delta}{2+\delta}P(y=0)$, then
		\begin{align}
			|P'(x|y=0)-P(x|y=0)| <\delta
		\end{align}
	\end{lemma}
	\begin{proof}
		We can lower bound $P'(y=0)$ as 
		\begin{align}
			P'(y=0) > P(y=0)  - \|P'-P\|_{1} > P(y=0)\frac{2}{2+\delta}
		\end{align}
		Then, we have
		\begin{align}
			|P'(x|y=0)-P(x|y=0)| &= \left| \frac{P'(x,y=0)}{P'(y=0)} - \frac{P(x,y=0)}{P(y=0)}\right|\\
			&\leq  \left| \frac{P'(x,y=0)}{P'(y=0)} - \frac{P(x,y=0)}{P'(y=0)}\right| + \left| \frac{P(x,y=0)}{P'(y=0)} -\frac{P(x,y=0)}{P(y=0)} \right| \\
			&\leq  \frac{1}{P'(y=0)}\left| P'(x,y=0) - P(x,y=0)\right| + P(x,y=0)\left| \frac{1}{P'(y=0)} -\frac{1}{P(y=0)} \right| \\
			&\leq  \frac{\|P'-P\|_{1}}{P'(y=0)} + P(x,y=0)\left| \frac{P(y=0)-P'(y=0)}{P'(y=0)P(y=0)}  \right| \\
			&\leq \frac{\|P'-P\|_{1}}{P'(y=0)}\left( 1 +  \frac{P(x,y=0)}{P(y=0)}   \right) \\
			&\leq \frac{2\|P'-P\|_{1}}{P'(y=0)}\\
			&\leq \frac{2\delta }{(2+\delta)}\frac{P(y=0)}{P'(y=0)}\\
			&< \delta
		\end{align}
		Where we have used the fact that $P(x,y=0) \leq P(y=0)$, and substituted expressions for $\|P'-P\|_{1}$ and $P'(y=0)$ in terms of $P(y=0)$ in the final steps.
	\end{proof}

	\begin{definition}(PostBQP)
		A language $L$ is in the class $\PostBQP$ iff there exists a uniform family of postselected circuits $\{C_w\}$ with a decision port $x$ and a postselection port $y$, and 
		\begin{align}
			\text{if } w\in L, P_{C_w}(x=1|y=0) \geq 1/2 + \delta\\
			\text{if } w \not\in L, P_{C_w}(x=1|y=0) \leq 1/2 - \delta
		\end{align} 
		where $\delta$ can be chosen arbitrary such that $0 < \delta < 1/2$. Further, the class $\PostBQP^*$ is defined with the added restriction that $P_{C_w}(y = 0) > 2^{-6m-4}$, where $m$ is the number of qubits, and \cite{fujii2016computational} show that $\PostBQP^* = \PostBQP = \PP$.
	\end{definition}

	\begin{lemma} \label{Lemma:InverseExponential} (Restatement of \cref{Lemma:Hardness_Exponential})
		There exists a family of IQP circuits $\mathcal{C}$, constructed on a 3D cubic lattice, consisting of a single layer of $e^{iZ\pi/8}$ gates and 4 layers of nearest-neighbour $e^{iZZ\pi/4}$ gates, such that sampling from any probability distribution $Q$ over bitstrings which satisfies $\|P_{C,q}-Q\|_{1} < (1-q)^n2^{-6n-4}/5$ is not possible with a classical polynomial algorithm, unless the polynomial hierarchy collapses to the third level, when $q < 0.134$.
	\end{lemma}
 
	\begin{proof}
 
		First, we will outline the proof in \cite{fujii2016computational} for review, then we will make our modifications. As in \cite{fujii2016computational}, by topologically protected MBQC, any quantum circuit can be encoded into a fault-tolerant IQP circuit of the form mentioned in \cref{Lemma:InverseExponential}, such that postselection recovers the original output distribution. 
		Explicitly, for any quantum circuit $C_w$ on $m$ qubits and error parameter $\epsilon$, one can construct $C$ on $n = \poly(m,\log (1/\epsilon))$ qubits, such that
		\begin{align}
			\|P_{C,q|y'=0} - P_{C_w}\|_{1} \leq \epsilon
		\end{align}
		where $y'$ is an error-detection register. This is essentially a variant of the threshold theorem, for the case of postselected error detection, rather than error correction (instead of doing fault-tolerant MBQC which requires adaptive measurement, we post-select on outcomes corresponding to `no error-detected').

		Suppose we choose some uniform family of quantum circuits $C_w$ which decide some $\PP$-complete language $L$ with confidence gap $0<\delta<1/2$, with decision port $x$ and postselection port $y$ where $P_{C_w}(y=0) > 2^{-6m-4}$. 
		For some $0 < \delta' < \delta$, suppose we encode this circuit into an IQP circuit $C$ using the fault-tolerance construction mentioned above with polynomial overhead, where $\epsilon  = 2^{-6m-4}\frac{\delta'}{2+\delta'}$. 
		By \cref{Lemma:postselect} and using the fact that $P_{C_w}(y=0) > 2^{-6m-4}$, this means that $|P_{C,q}(x|y=0,y'=0) - P_{C_w}(x|y=0)| < \delta'$. 
		Clearly, post-selection on \jw{$y,y'$ registers for} $P_{C,q}$ can also decide $L$, with confidence gap $\delta-\delta'$. 
		There are standard complexity-theoretic reductions to show hardness of \textit{exactly} sampling from probability distributions that decide $PP$-complete problems under postselection (see \cite{hangleiter2206computational} for review). 
		This concludes the outline of the proof from \cite{fujii2016computational}. 
		
		Now, note that  $P_{C,q}(y'=0) \geq (1-q)^n$ because $y'=0$ in $P_{C,q}$ when no bit-flip error occurs. This means that
		\begin{align}
			P_{C,q}(y=0,y'=0) &= P_{C,q}(y=0|y'=0)P_{C,q}(y'=0)\\
			&\geq P_{C,q}(y=0|y'=0)(1-q)^n\\
			&\geq (P_{C_w}(y=0) - \epsilon)(1-q)^n\\
			&> 2^{-6m-4}\left(1-\frac{\delta'}{2+\delta'}\right)(1-q)^n\\
			&> 2^{-6m-4}\frac{2}{2+\delta'}(1-q)^n
		\end{align}
		For some $0 < \delta'' < \delta'$, suppose we consider some distribution $Q$ such that $\|P_{C,q} - Q\|_{1} \leq (1-q)^n2^{-6m-4}\frac{2}{2+\delta'}\frac{\delta''}{2+\delta''}$. Using the triangle inequality and two applications of \cref{Lemma:postselect}, we have that
		\begin{align}
			|Q(x|y=0,y'=0) - P_{C_w}(x|y=0)| &\leq |Q(x|y=0,y'=0) - P_{C,q}(x|y=0,y'=0)| \\
			&+ |P_{C,q}(x|y=0,y'=0) - P_{C_w}(x|y=0)|\\
			&< \delta'' + \delta'
		\end{align}
		
		\noindent Thus, choosing $\delta''$ to be close to $1/2$ and $\delta'$ to be close to $0$, one can solve $\textsf{PP}$-complete problems by postselecting on $Q$ whenever $\|P_{C,q} - Q\|_{1} \leq (1-q)^n2^{-6n-4}/5$, which implies that sampling from $Q$ is hard (assuming $PH$)
	\end{proof}

 \section{Equivalence of Gibbs States and Circuits with Input Noise} \label{Appendix:Gibbs_IQP_Equivalence}

  Here we provide a proof of \cref{Lemma:Gibbs_to_Circuit_Sampling_Main} in the main text.
	\begin{lemma}\label{Lemma:Gibbs_to_Circuit_Sampling}
		For any circuit $C$ constructed from $k$-local gates of depth $d$, there exists a $O(kd)$-local parent Hamiltonian $H_C$, such that:
  \begin{align*}
      \frac{1}{Z}e^{-\beta H_C} &= C(\cD_q^{\ox n}(\ket{0}\bra{0}))C^\dagger
  \end{align*}
  for $q = \frac{e^{-\beta}}{1+e^{-\beta}}$.
	\end{lemma}

	\begin{proof}
		Let $C$ denote an arbitrary quantum circuit.
		Define a corresponding parent Hamiltonian $H_C = C\HNI C^\dagger$ with corresponding Gibbs state:
		\begin{align}
			\frac{1}{Z}e^{-\beta H_C} &= \frac{1}{Z}e^{-\beta \sum_x \lambda_x C\ket{x}\bra{x}C^\dagger} \\
			&= \frac{1}{Z} \sum_{x\in \{0,1\}^n} e^{-\beta  \HW(x) }C\ket{x}\bra{x}C^\dagger \label{Eq:Sample_Distribution}
		\end{align} 
		Note that because unitary transformations preserve the spectrum, the partition functions of $\HNI$ and $H_C$ are the same. Let $|x|$ be hamming weight of bitstring $x \in \{0,1\}^n$.
        \noindent We can then compare this to the output state of $C$ with a single set of bit-flip noise of strength $q$ occurring on the input.
        \begin{align}
            C(\cD_q^{\ox n}(\ket{0}\bra{0}))C^\dagger &= C \left(\sum_{x \in \{0,1\}^n} q^{\HW(x)}(1-q)^{n-\HW(x)}\ket{x}\bra{x}^{\ox n} \right)C^\dagger  \\
            &= \sum_{x \in \{0,1\}^n} \bigg(\frac{q}{1-q}\bigg)^{\HW(x)}(1-q)^{n} C\ket{x}\bra{x}^{\ox n} C^\dagger \label{Eq:Noisy_Circuit}
        \end{align}
        Comparing \cref{Eq:Sample_Distribution} to \cref{Eq:Noisy_Circuit}, and noting that $Z = (1+e^{-\beta})^n$, we can see that if we set $\frac{q}{1-q} = e^{-\beta}$ and $Z = (1+e^{-\beta})^n = (1+\frac{q}{1-q})^n = 1/(1-q)^n$, then these two states are the same.
        From here it is straightforward to see that the distribution $ P(s) \coloneqq \bra{s}\frac{e^{-\beta H_C}}{Z}\ket{s}$ satisfies:
        \begin{align*}
            P(s) = P_{C, q }(s)  
        \end{align*}
        for  $q = \frac{e^{-\beta}}{1+e^{-\beta}}$.
	\end{proof}

\section{CNOT Encoding }\label{Appendix:CNOT_Encoding}

We use a result from \citeauthor{Bremner_2017} \cite{Bremner_2017} which shows that IQP circuits can be made robust to bit-flip noise on the input distribution, at the cost of non-locality in the gate set. 
The fundamental idea is to encode a repetition code into a given hard-to-sample IQP circuit, and then decode with high probability. 

We provide a brief explanation of the encoding as follows. Suppose we wish to encode some logical IQP circuit $C$ into a physical circuit $C_{r}$. For each qubit $i$ involved in $C$ (and initialised in the $\ket{0}$ state), prepare a block of $r$ qubits labelled $i_1,\ldots, i_r$ in $C_{r}$, which are initialised to the $\ket{0}$ state. Now, every $2$-qubit diagonal gate between qubits $i$ and $j$ in $C$ can be written in the form $D = e^{i \theta_1 Z_i + \theta_2 Z_j + \theta_3 Z_iZ_j}$ (modulo global phase), for some $\theta_1,\theta_2,\theta_3$. This notation is known as an `X-program' \cite{shepherd_2009, shepherd2010binary, shepherd2010quantum,  Bremner_2017, gross2023secret}. 
The encoded circuit $C_{r}$ then involves the following transformation for each $2$-qubit diagonal gate
\begin{align}
    D = e^{i \theta_1 Z_i + \theta_2 Z_j + \theta_3 Z_iZ_j} \to D_{enc} = e^{i \theta_1  Z_{i_1}Z_{i_2}\ldots Z_{i_r} + \theta_2 Z_{j_1}Z_{j_2}\ldots Z_{j_r} + \theta_3 Z_{i_1}Z_{j_1}Z_{i_2}Z_{j_2}\ldots Z_{i_r}Z_{j_r}} \label{Eq:Transform}
\end{align}
In \citeauthor{Bremner_2017}, it is shown that this construction essentially encodes each logical output bit into $r$ physical output bits, each with an independent probability $q$ of being incorrect. 
Therefore, the logical output distribution can be recovered by taking a majority vote. 
This is captured in the following lemma.

\begin{lemma} (From \cite{Bremner_2017}) \label{Lemma:BMS}
Let $C$ be an arbitrary IQP circuit constructed with $2$-qubit gates of depth $d$ on $n$ qubits. Then, for any integer parameter $r \geq 1$, there is an encoded IQP circuit $C_{enc}$ constructed with $2r$-local gates of depth $d$ on $nr$ qubits, and a decoding algorithm $A$ such that,
\begin{align}
    A(P_{C_{enc},q}) = P_{C,\pfail}
\end{align}
where $\pfail \leq (4q(1-q))^{r/2}$.
\end{lemma}

It can be seen that the parent Hamiltonian of the above construction $H_{C_{enc}}$ will be $O(r)$-local.
Now, we show that this encoding can be achieved without non-local gates by using a networks of $CNOT$s
\begin{lemma}\label{Lemma:CNOT}
    Suppose IQP circuit $C$ on $n$ qubits is encoded into an encoded circuit $C_{enc}$ on $nr$ qubits as per \cref{Lemma:BMS}.
    For each logical qubit $i$ of $C$, let the block of $r$ physical qubits be labelled $i_1,\ldots,i_r$. 
    Let $C_{1}$ denote the circuit $C$ applied transversally\footnote{Transversally here means that for each gate between a pair of qubits $(i,j)$, in the encoded circuit the same gate now acts between the pair of qubits $(i_1, j_1)$.} to the first qubits of each block ($1_1,2_1,\ldots,n_1$). 
    Let the unitary $B$ be composed of the following CNOT network,
\begin{align}
    B =  \prod_{i \in 1,\ldots,n}\prod_{j \in 2,\ldots,r}CNOT_{i_1,i_j}
\end{align}
    where $CNOT_{x,y}$ is a CNOT controlled on qubit $x$ and targeting qubit $y$. 
    Then,
    \begin{align}
        B^\dagger C_{1}B = C_{enc}
    \end{align}
\end{lemma}
\begin{proof}
Suppose $C$ is given by the diagonal gates $D_1,\ldots,D_m$, so that $C =  H^{\otimes nr} \prod_{l \in 1,\ldots,m} D_l H^{\otimes nr}$. Suppose $C_{enc}$ is given by the diagonal gates $D_{1}',\ldots,D_{m}'$, so that the unitary produced by $C_r$ is $C_{enc} = H^{\otimes nr} \prod_{l \in 1,\ldots,m} D_{l}' H^{\otimes nr}$. 

Define $B_H =H^{\otimes n} B H^{\otimes n}$. By a circuit identity, conjugating a CNOT with Hadamards on either qubit reverses the target and the control (i.e. $(H_x \otimes H_y)CNOT_{x,y}(H_x \otimes H_y) = CNOT_{y,x}$). Therefore,
\begin{align}
    B_H = \prod_{i \in 1,\ldots,n}\prod_{j \in 2,\ldots,r}CNOT_{i_j,i_1}
\end{align}

Due to another circuit identity that $CNOT_{x,y}(\ident_x \otimes Z_y)CNOT_{x,y} =Z_x \otimes Z_y$, we can establish that $B_H$ has the following property,
\begin{align}
    B_H^\dagger Z_{i_1}B_H = \prod_{j =1,\ldots,b} Z_{i_j}
\end{align}
Suppose we use $D_{l,1}$ to denote gate $D_l$ applied transversely to the first qubits of each block ($1_1,2_1,\ldots,nr_1$). By examining \cref{Eq:Transform} and comparing it the above, one can see that
\begin{align}
    B_H^\dagger D_{l,1}B_H = D_{1}'
\end{align}

Therefore, we have,
\begin{align}
    B^\dagger C_{1} B &= B^\dagger \bigotimes_{i \in 1,\ldots,n} H_{i_1} \prod_{l \in 1,\ldots,m} D_{l,1} \bigotimes_{i \in 1,\ldots,n} H_{i_1} B \\
    &= B^\dagger H^{\otimes nr} \prod_{l \in 1,\ldots,m} D_{l,1} H^{\otimes nr} B \\
    &=  H^{\otimes nr} B_H^\dagger \prod_{l \in 1,\ldots,m} D_{l,1} B_H H^{\otimes nr}\\
    &= H^{\otimes nr} \left( \prod_{l \in 1,\ldots,m} B_H^\dagger D_{l,1} B_H\right) H^{\otimes nr}\\
    &= H^{\otimes nr} D_{l}'  H^{\otimes nr}\\
    &= C_{enc}
\end{align}
where in the second step we have used the fact that $H H = I$ (so we can add a pair of Hadamards to every qubit without a Hadamard).
\end{proof}

Ref. \cite{bergamaschi2024sample}  use a similar encoding scheme to ours. 
Specifically, they construct circuit $C' =  B'C_1$, where $B'$ is a different $CNOT$ network than $B$, but also acting amongst blocks of $r$ qubits. 
By examining $B'$, a simple adaptation of the above proof will show that $B' C_1 B'^\dagger = C_{enc}$. This means that the encodings of Ref. \cite{Bremner_2017}, Ref.\cite[Lemma 8.1]{bergamaschi2024sample}, and \cref{Lemma:CNOT} are all exactly equivalent up to CNOT networks on the output state. 
Since CNOTs commute with measurement, this means their output distributions are equivalent modulo post-processing. This suggests that further improvements to the noise-robustness of these constructions might require changes to the underlying error-correction code (the repetition code) rather than circuit-level optimisations. We also use the fact that CNOTs commute with measurement to establish the following lemma.
\begin{lemma} \label{Lemma:CNOT_Encoding}
Let $C$ be an arbitrary IQP circuit constructed with $2$-qubit gates of depth $d$ on $n$ qubits. Then, for any integer parameter $r \geq 1$, there is an encoded IQP circuit $C^*$ constructed with $2$-qubit gates of depth $d+r$ on $nr$ qubits, and a decoding algorithm $A$ such that,
\begin{align}
    A^*(P_{C^*,q}) = P_{C,\pfail}
\end{align}
where $\pfail \leq (4q(1-q))^{r/2}$. Furthermore, the parent Hamiltonian of $C^*$ has locality $k \leq d+2$ and degree $\Delta \leq r(d+1)$
\end{lemma}

\begin{proof}
    Using the notation of \cref{Lemma:CNOT}, define $C^* = C_1 B$.  $B^\dagger$ is composed of classical gates, which means it commutes with measurement. Therefore, we can apply it on the output distribution of $C_1 B$, and then perform the decoding algorithm $A$ of \cref{Lemma:BMS}. That is,
    \begin{align}
        A^*(P_{C^*,q}) = A(B P_{C^*,q}) = A(P_{BC^*,q}) = A(P_{C_r,q}) = P_{C,\pfail}
    \end{align}

    To calculate the locality, we use the following circuit identities.
    \begin{align}
    CNOT_{x,y} (\ident_x \ox Z_y)CNOT_{x,y} &= Z_x\ox Z_y \\
    CNOT_{x,y} (Z_x\ox \ident_y)CNOT_{x,y} &= Z_x\ox \ident_y
\end{align}
We first conjugate $H_{NI}$ with $B$
    \begin{align*}
      &BH_{NI}B^\dagger \\
      &= B \sum_{i \in [n],j \in [r]} 
     \frac{\ident-Z_{i_j}}{2} B^\dagger\\
     &=  \sum_{i \in [n]}  B\frac{\ident-Z_{i_1}}{2}B^\dagger  + \sum_{i \in [n],j \in [2,r]} B\frac{\ident-Z_{i_j}}{2} B^\dagger  \\
      &=  \sum_{i \in [n]}\prod_{j \in [2,n]} CNOT_{i_1,i_j} \frac{\ident-Z_{i_1}}{2}\prod_{j \in [2,n]} CNOT_{i_1,i_j} + \sum_{i \in [n],j \in [2,r]} CNOT_{i_1,i_j} \frac{\ident-Z_{i_j}}{2}CNOT_{i_1,i_j}  \\
      &= \sum_{i \in [n]} \frac{\ident -Z_{i_1}}{2}+ \sum_{i \in [n],j \in [2,r]}  \frac{\ident -Z_{i_1} Z_{i_j}}{2}  
    \end{align*}
    Note that for each qubit $i_1$, there are at most $d$ diagonal gates acting on qubit $i_1$ in $C_{1}$. Thus, when conjugating the $Z_{i_1}Z_{i_j}$ term with $C_1$, the resulting interaction term spans at most $d+2$ qubits, so $k \leq d+2$. 

    Each qubit $i_1$ appears in all $r-1$ terms of the form $C_1 Z_{i_1}Z_{i_j}C_1^\dagger$. It also appears in all $r-1$ terms of the form $C_{1}Z_{i'_1}Z_{i'_j}C_{1}^\dagger$ for each of the $d$ neighbours $i'$. If we count the $Z_{i_1}$ terms as well, the total degree is $\Delta \leq (d+1)r$
\end{proof}

\end{document}